\documentclass{article}
\usepackage[affil-it]{authblk}
\usepackage{graphicx}
\usepackage[space]{grffile}
\usepackage{latexsym}
\usepackage{textcomp}
\usepackage{longtable}
\usepackage{multirow,booktabs}
\usepackage{amsfonts,amsmath,amssymb}
\usepackage{url}
\usepackage{hyperref}
\hypersetup{colorlinks=false,pdfborder={0 0 0}}
\newif\iflatexml\latexmlfalse

\usepackage[english]{babel}

\usepackage[utf8]{inputenc}
\usepackage{amsmath}
\usepackage{amsfonts}
\usepackage{amssymb}
\usepackage{mathrsfs}
\usepackage{amsthm}

\numberwithin{equation}{section}

\newcommand{\A}{\mathcal{A} }
\newcommand{\C}{\mathcal{C} }

\newcommand{\Si}{\mathcal{S} }

\newcommand{\T}{\mathsf{True} }
\newcommand{\F}{\mathsf{False} }

\newtheorem{theorem}{Theorem}[section]
\newtheorem{definition}{Definition}[section]
\newtheorem{corollary}{Corollary}[section]
\newtheorem{remark}{Remark}[section]

\begin{document}

\title{Several Proofs of Security for a Tokenization Algorithm}


\author{Riccardo Longo\footnote{\textit{riccardo.longo@unitn.it}} , Riccardo Aragona\footnote{\textit{riccardo.aragona@unitn.it}} , and Massimiliano Sala\footnote{\textit{maxsalacodes@gmail.com}}}
\affil{University of Trento}
  
%

\date{\today}

\maketitle

\begin{abstract}
In this paper we propose a tokenization algorithm of  Reversible Hybrid type, as defined in PCI DSS guidelines for designing a tokenization solution, based on a block cipher with a secret key and (possibly public) additional input. We provide some formal proofs of security for it, which  imply our algorithm satisfies the most significant security requirements described in PCI DSS tokenization guidelines. Finally, we give an instantiation with concrete cryptographic primitives and fixed length of the PAN, and we analyze its efficiency and security.%
\end{abstract}

\section{Introduction}
In recent years, credit cards have become one of the most popular payment instruments.
Their growing popularity  has brought
many companies to store the card information of its customers to make
simpler subsequent payments. This need is shared by many other actors in the payment process.
On the other hand, credit card data are very sensitive information, and then the theft of such data is considered a serious threat.

Any company that stores credit card data aims to achieve the \emph{Payment Card Industry Security Standard
Council} (\emph{PCI SSC}) compliance. The PCI SSC is an organization, founded by the largest payment card networks, which has developed several standards and recommendations. One of these is called the \emph{PCI Data Security Standard} (\emph{PCI DSS} \cite{pcidss}) and its goal is to guarantee the security of credit card data. PCI DSS requires that companies that handle payment cards protect the data of the cardholder when these are stored, transmitted or processed.

These stringent requirements  led  to consider a new method for storage and transmission of the card information: instead of protecting the actual card data, it is easier to remove them (when their storing is not requested) and replace them with another
value, called  \emph{token}.
Tokens are alpha-numeric strings representing the \emph{PAN} (\emph{Primary Account Number}) of a payment card, and that may have a format similar to it.
In any case, from a token it must be infeasible (without additional information) to recover the PAN from which it was generated.
This process is called \emph{tokenization}.
In \cite{Santiago2016}, the authors present an interesting formal cryptographic study of tokenization systems and their security.

In recent years, PCI SSC has drafted some guidelines to design a tokenization solution \cite{pcitoken11,pcitoken15}. In \cite{pcitoken15}, the following five types of tokens are described: \emph{Authenticatable Irreversible Tokens}, \emph{Non-Authenticatable Irreversible Tokens}, \emph{Reversible Cryptographic Tokens}, \emph{Reversible Non-Cryptographic Tokens} and \emph{Reversible Hybrid Tokens}.\\
In particular, a reversible tokenization algorithm , i.e., providing the possibility for entities using or producing tokens to obtain the original PAN from the token, can be designed in three ways:  
\begin{itemize}
\item \emph{Reversible Cryptographic}, if it  generates tokens from PANs using strong cryptography. In particular, a mathematical relationship between PAN and corresponding token exists. In this case, the PAN is never stored; only the cryptographic key is (securely) stored.
\item \emph{Reversible Non-Cryptographic}, if obtaining a PAN from its token is only by a data look-up in a dedicated
server (called a Card Data Vault). In this case, the token has no mathematical relationship with its associated PAN and the only thing to be kept secret is the actual relationship between the PAN and its token (e.g., a look-up table in the Card Data Vault).
\item A reversible tokenization algorithm is called \emph{Hybrid} if it contains some features of both Reversible Cryptographic tokens and Reversible Non-Cryptographic tokens. A typical situation of this type is when, although  there is a mathematical relationship 
between a token and its associated PAN, a data look-up table must be used to retrieve the PAN from the token.
\end{itemize}

In this paper we propose a tokenization algorithm of the \emph{Reversible Hybrid} type, based on a block cipher with a secret
key and (possibly public) additional input.
We provide some formal proofs of security for it. To fully appreciate our design and the proposed proofs
it is necessary to analyze the PCI requirements in more detail.
The remainder of the paper is thus as follows:
\begin{itemize}
\item In Section \ref{requirements}, we analyze some  PCI requirements for a tokenization algorithm 
      (\cite{pcitoken11,pcitoken15}).
\item In Section \ref{alg}, we describe our tokenization algorithm;
\item In Section \ref{proof}, we prove the security of the algorithm defined in Section \ref{alg}  in a very general scenario, which would imply our       algorithm satisfies most requirements present in Section \ref{requirements}. More precisely, we 
      present a security notion, Indistinguishability under a Chosen-Plaintext Attack (IND-CPA),
      for a tokenization algorithm of our type, and we prove it under the assumption
      that our core cryptographic algorithm satisfies a standard IND-CPA. We also provide a separate proof for a special requirement.
\item In Section \ref{sec:con}, we give an instantiation of  our algorithm, considering concrete cryptographic primitives and fixing PAN length, and analyzing the security and efficiency in this real-life application.           
\end{itemize}

\section{Requirements} \label{requirements}

The first two requirements are not linked to security: 
\begin{itemize}

\item Although  tokens can enjoy a variety of formats, the most convenient is probably the same format of the PAN itself,
since in this case a token can move inside a payment network and also be used as a payment token
(for a definition of a \emph{payment token} see p. 13 in \cite{emv14}).
But if we try to create a token by a direct encryption of the PAN we will meet several problems, since  the output of the encryption may not have a format like a PAN.
Indeed, we must keep in mind that PCI requests the use of standard encryption algorithms and so we cannot create ad-hoc cryptographic primitives, but we must rely on established algorithms. \\
So we must face the problem of designing algorithms that preserve the message format, or the so-called \emph{Format Preserving Encryption} (\emph{FPE}) \cite{Bellare2009}.
In literature, there are some interesting examples of algorithms that solve such problem \cite{bellareffx,brier2010bps,Hoang2012,Morris2009,stefanov2012fastprp}.

\item It must be possible to obtain different tokens from a single PAN (even one per transaction, if necessary),
      so the tokenization algorithm will require additional inputs, such as, a transaction counter, an expiration date, etc. .

\end{itemize}

Concerning security issues, there are many security requirements that our algorithm has to satisfy in order to
meet PCI compliance.
We list here the main requests:
\begin{enumerate}
\item[A1] ``\emph{the recovery of the original PAN must not be computationally feasible knowing only the token or a number of tokens.}" (p. 6 in \cite{pcitoken11}).\\
In other words, even if an attacker has managed to collect many tokens, all coming from the same PAN, possibly
even on a long period of time, they must be computationally unable to retrieve the corresponding PAN.
This is a form of ciphertext-only attack.

\item[A2] ``\emph{access to multiple token-to-PAN pairs should not allow the ability to predict or determine other PAN values from knowledge of only tokens.}" (p. 6 in \cite{pcitoken11} and GT4 in \cite{pcitoken15}). \\ This is a known-plaintext attack.

\item[A3] ``\emph{Tokens should have no value if compromised or stolen, and should be unusable to an attacker if a system storing only tokens is compromised}'' (p. 6 in \cite{pcitoken11}). \\
Since this sentence comes immediately after A1 and A2, which aim at preventing PAN recoveries,  we take the goal
of this rather cryptic sentence to be the prevention of unauthorized token generation. In other words, an attacker possessing many tokens (but not knowing the corresponding PAN's) must be unable to generate even \emph{one}  other valid token. This condition is drastically different from A1, because here we do not require the attacker to be able to deduce any of the involved PAN's.
However, there are two matters. First, since they needs to compute \emph{valid} tokens, they must  have control
on any other input of the tokenization algorithm (such as, a transaction counter). Second, they needs to know
for which PAN they can generate a token.  
       
\item[A4] ``\emph{Converting from a token produced under one cryptographic key to a token produced under another cryptographic key  should require an intermediate PAN state---i.e., invocation of de-tokenization.}"  (GT 11 in \cite{pcitoken15}).\\
Since our system uses a block cipher with additional (public) input, we take this to mean that if an attacker gets a token
obtained by a PAN, a secret key and an additional input, then they must be unable to compute any token corresponding 
to the same PAN (and same additional input) but to a different key.

\item[A*] ``\emph{The recovery of the original PAN should be computationally infeasible knowing only the token, 
       a number of tokens, or a number of PAN/token pairs}" (GT 5 in \cite{pcitoken15}).\\
       This is a repetition of A1 and A2.
       
\end{enumerate}

To prove that our tokenization algorithm satisfies A1, A2 and A3, we will prove in Section \ref{proof} that it satisfies an even stronger condition, under the assumption of the strength of the core primitive we are using to define it (a block cipher). \\
Requirement A4 requires a separate proof in Section \ref{proof}.

\section{Algorithm} \label{alg}

A card number is formed by three concatenated parts: the IIN (also called BIN), that identifies the card Issuer, a numeric code, that
identifies the account, and a check digit. 
We assume to replace the IIN with another fixed code (called a ``token BIN" in \cite{emv14}, p. 14), which marks the resulting card number as a token, so we will ignore the first part in the description of our algorithm. 
We will also compute the check digit as appropriate, so we can discard it too in the forthcoming formal description of our algorithm.

\indent We assume to be able to invoke the encryption function of a block cipher $E$, just by sending a plaintext and obtaining a ciphertext, with a key that is kept somewhere protected and that we do not need to know.  
We can think of $E$ as the first \emph{ingredient} of our algorithm.
$\mathbb{K}$ denotes the keyspace of $E$ and $K\in\mathbb{K}$ will be any key, so that we can view $E$ as a function $E: \mathbb{K} \times (\mathbb{F}_2)^{m} \longrightarrow (\mathbb{F}_2)^{m}$ for some $m\in \mathbb{N}$.
With $K$ fixed, $E$ is a permutation acting on the set $(\mathbb{F}_2)^{m}$ of the $m$-bit strings. With standard block ciphers we have $m=64$  or $m=128$. We assume as usual that the set of its encryption functions forms a random sample of the set of permutations acting on $(\mathbb{F}_2)^{m}$.\\
For strings we use a notation like $|0110|_2$, where the index $2$ denotes that we are
using only symbols from $\{0,1\}$, i.e., remainders of division by $2$.

\indent Our algorithm processes two inputs: a numeric code coming from the PAN and an additional input.
\begin{itemize}

\item By \emph{numeric code} we mean a string of $\ell$ decimal digits, $\ell$ any agreed number $\ell\in \mathbb{N}$, and we formally define the set of our numeric codes as $\mathbb{P}:=\{0,1,\ldots,9\}^\ell$  At present, $13 \leq\ell\leq 19$ for numeric codes coming from PANs \cite{emv14}, but we do not need to impose any limitation.
We observe that $\mathbb{P}$ is in obvious bijection with the integer set $\{a \in \mathbb{N} \mid 0\leq a < 10^\ell\}$,
and so we can view a numeric code also as a non-negative integer, but care has to be
taken to pass from one representation to the other. Let $[y]_b^s$ denote the representation (string) of 
a positive integer $y < b^s$ in base $b$ with $s$ digits, where the most significant digits are on the \emph{left}.\\
For example, $\quad[12]_{10}^{2}=|12|_{10}$, $[12]_{10}^{3}=|012|_{10}$ and $[13]_2^5=|01101|_2$. \\
Then any positive integer $X$ such that $X <10^\ell$  can be easily converted to $[X]_{10}^\ell \in \mathbb{P}$. We will use a bar to denote the conversion from a string to a number,
like $\overline{|12|_{10}}=12$ and $\overline{|01100|_2}=12$.

\item The role of the additional input is to allow for different tokens corresponding to the same PAN
      and so it can be anything, such as a transaction counter or an integer denoting the current time. 
      Formally, we will identify it as a binary string of finite but arbitrary length (as for example the binarization of a transaction counter). We will call $ \mathbb{U}$ the set containing all these strings, so that any $u\in  \mathbb{U}$ is implicitly meant to be an additional input.
\end{itemize}

\indent Let $n:=\lceil \log_2(10^\ell)\rceil$ be the maximum number of bits required to represent a number with $\ell$ decimal digits. Since most of the block ciphers used in real life applications have a block-size of at least 64 bits, and for the maximum length of a PAN $\ell = 19$ we have $n = 64$, we assume that $\ell$ is such that $n \le m$.\\
Now we can present $f$, the second \emph{ingredient} of our algorithm, which is a public function
 $$
    f:\mathbb{U} \times \mathbb{P} \longrightarrow (\mathbb{F}_2)^{m-n}
 $$ 
 In other words, given an additional input $u\in\mathbb{U}$  
 and a numeric code $X\in\mathbb{P}$ coming from the PAN , $f$ returns a string of $m - n$ bits.\\
 We require $f$ to be \emph{collision-resistant}, that is, it must be computationally
infeasible to obtain two distinct pairs $(X_1,u_1)$ and $(X_2,u_2)$ such that
$f(u_1,X_1)=f(u_2,X_2)$.
This requirement compels the image space to have a size large enough to prevent brute-force collision attacks.
So in the case of PAN tokenization we need to consider only block ciphers with block size of at least 128 bits, in order to have an image space of dimension at least 64 bits.
The purposes of this function $f$ are the followings:
\begin{itemize}
\item to pad the input of the tokenization algorithm to match the block size of the cipher;
\item to allow the creation of multiple different tokens from the same PAN using the same key, useful for example to change token for each transaction.
\end{itemize}
The output of $f$ could be seen as a \emph{tweak} in the context of Tweakable Encryption \cite{liskov2002tweakable}.
An example for this $f$ could be a truncated version of a cryptographic hash function.

\indent The third \emph{ingredient} for our tokenization algorithm is a database stored somewhere
securely that contains a look-up table of PAN-token pairs. Once we have generated a token, we assume it is
inserted in the table. However, to generate a new token, we need to access the database only via a function $\mathrm{check}$  that checks if the token is already stored and returns either $\T$ or $\F$ accordingly. 

One of the goals of a tokenization algorithm is to obtain an integer with $\ell$ decimal digits starting from another integer with the same length.
Since $n:=\lceil \log_2(10^\ell)\rceil < m$, we have to consider only a fraction of the output of $E$, in particular we will take the $n$ least significant bits of the output, and then convert this string back to an integer.
Given that $10^l < 2^n$, this integer could have $\ell + 1$ decimal digits.
To solve this problem we use a method known as \emph{Cycle Walking Cipher} \cite{black2002ciphers}, designed to encrypt messages from a space $\mathcal{M}$ using a block cipher that acts on a space $\mathcal{M'} \supset \mathcal{M}$, and obtain ciphertexts that are in $\mathcal{M}$.

\indent We are ready to write down our algorithm.\\
The \emph{Tokenization Algorithm} $T(K, X, u)$ executes the following steps:
\begin{enumerate}
\item $t:=f(u,X)\mid\mid \left[\bar{X}\right]_2^n$
\item $c:= E(K, t)$
\item if $(\bar{c} \mod 2^n) \geq 10^\ell$, then $t:=c$ and go back to step 2
\item $\texttt{token} := [ \bar{c} \mod 2^n]_{10}^\ell$
\item if ${\mathrm{check}}(\texttt{token}) = \T$, then $u:= u + 1$ and go back to step 1
\item return $\texttt{token}$
\end{enumerate}

The correctness of Algoritm $T$ is obvious, we now discuss its termination.

At Step 3 we check if $(\bar{c} \mod 2^n) \geq 10^\ell$. Since $x \mapsto E(K,x)$ is a random permutation,
we expect $c$ to resemble a random binary string in $(\mathbb{F}_2)^m$. Therefore, the number
$(\bar{c} \mod 2^n)$ is a random integer in $\{0,\ldots,2^n-1\}$. 
Recall that $2^{n-1} < 10^\ell < 2^n$. Therefore, the condition at Step 3 is met with probability $0 < p=\frac{2^n-10^\ell}{2^n} <1$. Going back to step 2 another pseudo-random
number is computed and the probability that the condition of Step 3 is not satisfied for it is again $p$.
Since the two events can be considered independent, due to the pseudorandomness property
of $E$, the probability of the joint event goes down to $p^2$, and so on.
Thefeore, the probability that the algorithm remains stuck at Step 3 is negligible.

At Step 5 we check if we already have the new token in our database . 
If we have it, we increase $u$ to $u+1$. We remain stuck only if
$f(u,X)=f(u+1,X)$, but this happens very rarely thanks to the collision resistance of $f$.

For a more detailed discussion of the probability to meet the conditions at Step 3 and at Step 5, see the instantiation of 
our algorithm given in Section \ref{sec:con}.

\section{Proof of Security}\label{proof}

In this section we will prove that the algorithm previously defined is secure in an Indistinguishability under a Chosen-Plaintext Attack scenario, under the condition that its core encrypting algorithm is secure in the same scenario. This will guarantee in particular A1, A2 and A3.
Then we will prove A4 separately.

\begin{definition}[IND-CPA] \label{cpa1}
    Let $E(K, m)\rightarrow c$ be an encrypting function.
    An \emph{Indistiguishability under Chosen Plaintext Attack} (IND-CPA) game for $E$ between an adversary $\A$ and a challenger $\C$ proceeds as follows: 
    \begin{description}
        \item[Phase I] $\A$ chooses a plaintext $m_i$ and sends it to $\C$, that responds with $c_i = E(K, m_i)$. This phase is repeated a polynomial number of times.
        
        \item[Challenge] $\A$ chooses two plaintexts $m^*_0, m^*_1$ (never chosen in Phase I) and sends them to $\C$ that selects $\nu \in \{0, 1\}$ at random and computes $c= E(K, m^*_{\nu})$. Then $\C$ sends $c$ to $\A$.
        \item[Phase II] Phase I is repeated, with the obvious restriction that $\A$ cannot choose $m^*_0$ or $m^*_1$.
        \item[Guess] $\A$ guesses $\nu' \in \{0, 1\}$. They wins if $\nu' = \nu$.
    \end{description}
    
    We say that the advantage $Adv_{\A}^E$ of $\A$ winning the IND-CPA game for the Encrypting algorithm $E$ is:
    \begin{equation*}
    Adv_{\A}^E= \left| Pr\left[\nu' = \nu \right] - \frac{1}{2}\right|
    \end{equation*}
    
    $E$ is said to be \emph{secure in a IND-CPA scenario} if there is no probabilistic polynomial-time algorithm $\A$ that wins the CPA game with more than negligible advantage.
\end{definition}

    We define a IND-CPA game for $T$ analogously to \ref{cpa1}, where messages are  replaced by numeric codes and additional inputs, while ciphertexts are replaced by tokens.
    So the adversary chooses two code/additional input pairs and tries to distinguish which of these  corresponds to the token returned by $\C$.
    The adversary is also able to request other tokens (corresponding to a polynomial number of pairs), that she may choose adaptively.\\
    \vspace{2mm}
\begin{definition}[IND-CPA for a Tokenization Algorithm] \label{cpa2}
   Let $T(K, X, u)\rightarrow \texttt{token}$ be a tokenization algorithm that takes as input a key $K$, a numeric code $X\in \mathbb{P}$  and an additional input $u\in \mathbb{U}$, and returns 
   a token $\texttt{token}$.
    An \emph{Indistiguishability under Chosen Plaintext Attack} (IND-CPA) game for $T$ between an adversary $\A$ and a challenger $\C$ proceeds as follows: 
    \begin{description}
        \item[Phase I] $\A$ chooses $(X_i, u_i)$ and sends it to $\C$, that responds with $\texttt{token}_i = T(K, X_i,  u_i)$. This phase is repeated a polynomial number of times.
        \item[Challenge] $\A$ chooses $(X^*_0, u^*_0), (X^*_1, u^*_1)$ (with $(X^*_0,u^*_0)$ and $(X^*_1,u^*_1)$ never chosen in Phase I) and sends them to $\C$ that selects $\nu \in \{0, 1\}$ at random and computes $\texttt{token}= T(K, X^*_{\nu}, u^*_{\nu})$, which sends to $\A$.
        \item[Phase II] Phase I is repeated, with the obvious restriction that $\A$ cannot choose $(X^*_0,u^*_0)$ or  $(X^*_1,u^*_1)$.
        \item[Guess] $\A$ guesses $\nu' \in \{0, 1\}$, she wins if $\nu' = \nu$.
    \end{description}
    
     We say that the advantage $Adv_{\A}^T$ of $\A$ winning the IND-CPA  game for the Tokenization Algorithm $T$ is:
    \begin{equation*}
    Adv_{\A}^T= \left| Pr\left[\nu' = \nu \right] - \frac{1}{2}\right|
    \end{equation*}
    
    $T$ is said to be secure in a IND-CPA scenario if there is no probabilistic polynomial-time algorithm $\A$ that wins the IND-CPA game with more than negligible advantage.
\end{definition}
\vspace{2mm}
\begin{theorem}[IND-CPA Security of Tokenization Algorithm]\label{thcpa2}
    Let $T$ be the Tokenization Algorithm described in Section~\ref{alg}, and let $E$ be 
    the block cipher used in Step 2 of $T$.
    If $E$ is secure in an IND-CPA scenario then $T$ is secure in an IND-CPA scenario.
\end{theorem}

\begin{proof}
    Let $\C$ be the challenger in the IND-CPA game for $E$, and $\A$ be an algorithm that can win the IND-CPA game for $T$ with more than negligible advantage $\epsilon$.
    We will build a simulator $\Si$ that plays the IND-CPA game for $E$ by simulating an IND-CPA game for $T$ and interacting with $\A$.
        
    For  Phase I and II we need to show how $\Si$ answers to tokenization queries $(X_i, u_i)$ made by $\A$.
    The function $f$ is publicly known, so $\Si$ may compute $t_i:=f(u_i,X_i)\mid\mid \left[\overline{X_i}\right]_2^n$ and queries $\C$ for the encryption of $t_i$, obtaining $c_i$ in response.
    If $(\overline{c}_i \mod 2^n) \geq 10^l$, then $\Si$ queries again $\C$, this time requiring the encryption of $c_i$, repeating this passage until $\C$ answers with $c_i$ such that $(\overline{c}_i \mod 2^n) < 10^\ell$
    (which will eventually happen
    since $c \mapsto E(K,c)$ is a random permutation).
    At this point $\Si$ answers to the query of $\A$ with $\texttt{token}_i :=  [ \overline{c}_i \mod 2^n]_{10}^l$. \\
Observe that the $c_i$'s, and hence the $\texttt{token}_i$'s, will be distinct with high probability, since $f$ is collision-resistant, even if the $X_i$'s are identical, as long as the pairs $(X_i,u_i)$'s are distinct.
    
    \indent In the challenge phase, $\Si$ receives from $\A$ two pairs $(X^*_0, u^*_0), (X^*_1, u^*_1)$.
    $\Si$ computes $t^*_j:=f(u^*_{j},X^*_{j})\mid\mid [\overline{X^*}_j]_2^n$ for $j \in \{0,1\}$ and submits them to $\C$ for the challenge phase of the IND-CPA game for $E$.
    $\C$ chooses at random $\nu \in \{0, 1\}$ and will respond with the challenge ciphertext $c$.
    If $(\overline{c} \mod 2^n) \geq 10^l$, $\Si$ queries  $\C$ requiring the encryption of $c$, repeating this passage until we have ${\texttt{token}}:= (\overline{c} \mod 2^n) < 10^\ell$ (which will eventually happen
    since $c \mapsto E(K,c)$ is a random permutation).
    At this point $\Si$ sends ${\texttt{token}}$ to $\A$ as the challenge token of the IND-CPA game for $T$.

    Eventually $\A$ will send to $\Si$ its guess $\nu'$ for which code has been tokenized into $\texttt{token}$.
    $\Si$ then forwards this guess to $\C$.
     It is clear that $\A$ guesses correctly if and only if $\Si$ guesses correctly since the simulation is seamless.
     
     Note that during the Challenge phase $\Si$ is not allowed to send to $\C$ messages submitted during Phase I, and in Phase II $\Si$ is not allowed to send to $\C$ the two messages submitted in the Challenge phase.
     Since the same restriction applies to the interaction between $\A$ and $\Si$, problems may arise only when $\Si$ queries for the re-encryption of ciphertexts to meet the condition $(\overline{c}_i \mod 2^n) < 10^\ell$.
     However the number of queries is polynomial and the encryption function $c \mapsto E(K,c)$ is a random permutation, so  the probability of such a collision is negligible.

     Finally, throughout the simulation Step 5 of the algorithm has been ignored, always supposing that the $\mathrm{check}$ function outputs $\F$ (to be coherent the $\mathrm{check}$ function  has to only check  the simulated tokens already generated by $\Si$).
     
     In phases I and II, when ${\mathrm{check}}({\texttt{token}}_i) = \T$, $\Si$ is able to follow the algorithm properly adjusting $u_i$ and querying $\C$.
     
     For the challenge phase instead we have to hope that ${\mathrm{check}}({\texttt{token}}) = \F$, which happens with non-negligible probability $\rho_{q_1}$ (this probability depends on the number of queries $q_1$ made in Phase I).
     When ${\mathrm{check}}({\texttt{token}}) = \T$, we cannot simulate correctly, so $\Si$ may directly guess randomly, and so the probability of guessing is $\frac{1}{2}$ in this case.
     Thus the advantage $\epsilon'$ of $\Si$ is:
     \begin{equation*}
        \epsilon' = \left|\left(\rho_{q_1} \left(\frac{1}{2} + \epsilon\right) + (1 - \rho_{q_1}) \frac{1}{2} \right) - \frac{1}{2}\right|
     \end{equation*}
     which is non-negligible since $\epsilon$ and $\rho_{q_1}$ are non-negligible.
     Thus $\Si$ has a non-negligible advantage winning the IND-CPA game for $E$.
\end{proof}

It is well-known that an encryption algorithm which is IND-CPA secure then it is secure against a \emph{known-plaintext attack}, and even more so against a \emph{ciphertext-only} attack.
We now show similar results for our tokenization algorithm by proving some of our claimed requirements.\\
\vspace{2mm}
\begin{corollary}[A1-A2-A3]\label{corcpa2}
    Let $T$ be the Tokenization Algorithm described in Section~\ref{alg}, and let $E$ be 
    the block cipher used in Step 2 of $T$.
    If $E$ is secure in an IND-CPA scenario then $T$ satisfies A1, A2 and A3.
\end{corollary}
\begin{proof}
\begin{itemize}
\item A1\\
      If $T$ does not satisfy A1 then an attacker $\A$ can obtain a PAN from a token
      once she gets enough tokens corresponding to the same PAN. So we can suppose that  $\A$ has access to an algorithm that takes in input a    
      polynomial number $N$ of $\texttt{token}$s and with a non-negligible probability outputs
      \begin{itemize}
        \item $X$, if all the $N$ $\texttt{token}$s correspond to $X$;
        \item $\F$, otherwise.
      \end{itemize}
     \indent $\A$ tries the IND-CPA game for $T$ and
      chooses $N-1$ pairs $(X,u_i)$, with the same $X$,  and sends them to $\C$, who responds with
        $\texttt{token}_i = T(K,X,u_i)$.\\
        
        Then $\A$ passes to the  Challenge Phase and chooses  $u^*_0$ and  $(X^*_1,u^*_1)$, with $X^*_1\ne X$ and sends $(X^*_0,u^*_0)$, with $X^*_0=X$, and 
        $(X^*_1,u^*_1)$ to $\C$ that selects $\nu\in \{0, 1\}$ at random and computes $\texttt{token}= T(K,X^*_{\nu}, u^*_{\nu})$. Then
        $\C$ sends $\texttt{token}$ to $\A$.\\
        
        Then $\A$ runs her algorithm passing as inputs $\{\texttt{token}_i\}_{1\leq i\leq N-1}$ and 
        $\texttt{token}$, obtaining with non-negligible probability either 
   $X_0^*=X$ or $\F$ and so she knows for sure and wins the IND-CPA game.

\item A2\\
If A2 does not hold, then $\mathcal{A}$ has access to an algorithm that takes two inputs, a polynomial number $N$ of token-to-PAN pairs and a token, and returns the PAN corresponding to the token with non-negligible probability.
Thanks to our convention in Section 3, we can assume the algorithm inputs
to be actually $N$ pairs of type $(X_i,{\texttt{token}}_i)$ 
(plus the single token ${\texttt{token}}^*$)
and the algorithm output to be the $X^*\in \mathbb{P}$ corresponding to 
$\texttt{token}^*$.

\indent The adversary tries the IND-CPA game for $T$ by choosing $N$ random  $(X_i,u_i)$'s 
and obtaining from $C$ their corresponding tokens $\texttt{token}_i$'s. 
\\
Then she passes to the Challenge Phase and sends two random pairs: $(X_0^*,u_0^*)$ and $(X_1^*,u_1^*)$. The Challenger returns the token $\texttt{token}$ corresponding to one of these. \\
Then $\mathcal{A}$ runs her algorithm passing as inputs the pairs $(X_i,\texttt{token}_i)$'s and $\texttt{token}$, obtaining with non-negligible probability the
correct $X\in \{X_0^*,X_1^*\}$, winning thus the IND-CPA game.  
 
\item A3\\
If A3 does not hold, then we can assume that $\mathcal{A}$ has access to an algorithm that takes as input a polynomial number $N$ of tokens and one additional input $u^*$, and that returns a pair $(X,\texttt{token}^*)$,
where $\texttt{token}^*$ is the token corresponding to $(X,u^*)$ with non-negligible probability.

\indent The adversary tries the IND-CPA game for $T$ by choosing $N$ random  $(X_i,u_i)$'s 
and obtaining from $C$ their corresponding tokens $\texttt{token}_i$'s. Then she chooses a $u^*\in \mathbb{U}$
and passes it along with these tokens to her algorithm, obtaining $X$ and $\texttt{token}^*$.
\\
Then she passes to the Challenge Phase and sends  pairs: $(X,u^*)$ and a random pair. The Challenger returns the token $\texttt{token}$ corresponding to one of these. \\
The adversary easily wins the game just by checking if $\texttt{token} = \texttt{token}^*$.
 
\end{itemize}
\end{proof}

\begin{theorem}[A4 for Tokenization Algorithm]\label{tha4}
Let $K,K^*\in \mathbb{K}$, $X\in \mathbb{P}$ and $u\in \mathbb{U}$.
If an attacker knows only $u$ and the token $\texttt{token}=T(K,X,u)$, then she is able to compute $\texttt{token}^*=T(K^*,X,u)$ only with negligible probability.
\end{theorem}
\begin{proof}
The two tokens come directly from the encryption of the same string $M=f(u,X) \mid\mid [\overline{X}]^n_2$ with two different
keys, except when the unlikely condition in Step 3 is met. So, with non-negligible probability, $\A$ is able to compute $E(K^*,M)$ from $E(K,M)$. This means that for a large portion of the plaintext space the two encryption functions are closely correlated, since from one it is possible to deduce the other without the need for decryption/re-encryption contradicting our first assumption
on $E$, that is, that  the set of its encryption function forms a random sample of the set of permutations acting on $(\mathbb{F}_2)^m$.
\end{proof}

\begin{remark}
The requirement GT5 in \cite{pcitoken15} ``\emph{The recovery of the original PAN should be computationally infeasible knowing only the token, a number of tokens, or a number of PAN/token pairs}'' directly follows  from A1 and A2. 
\end{remark}

 

\section{Conclusion}\label{sec:con}
Tokenization is a problem of practical interest, so we conclude giving an example of an instantiation with concrete cryptographic primitives and fixed length of the PAN, and we will analyze its efficiency and security.

Let us consider PANs with length $\ell=16$, so $n=54$, let us take AES-256 as the cipher $E$, and as the function $f$ we take SHA-256 truncated to $128-54=74$ bits.

\subsection{Security}
Given the results given in \cite{SHA15}, SHA-256 passed several statistical tests designed
to verify ``the absence of any detectable correlation between input and output and the
absence of any detectable bias due to single bit changes in the input string". Therefore, it can be considered collision-resistant and so it satisfies our purposes (see, for instance, Step 5 of our algorithm).\\ 
Moreoever, in \cite{nisttest}, the authors showed that AES-256 passed statistical tests designed to verify
the following properties:
\begin{itemize}
\item ``the absence of any detectable correlation between plaintext/ciphertext pairs
and the absence of any detectable bias due to single bit changes to a plaintext
block"; and
\item ``the absence of any detectable deviations from randomness".
\end{itemize}
Therefore, we can consider AES-256 IND-CPA secure and so it satisfies the hypothesis of  \ref{thcpa2} and  \ref{corcpa2}, and the randomness requirements of  \ref{tha4}.

\subsection{Efficiency}\label{subsec:eff}
The probability that the condition at Step 3 is met is
\begin{equation}\label{pstep3}
    p= \frac{2^{54} - 10^{16}}{2^{54}} \approx 0.445
\end{equation}
We have a geometric distribution, so the expected value of the number of iterations is:
$$
    \mathbb{E}_1 = \sum_{k=1}^{\infty} k p^{k-1} (1-p) \approx 1.801
$$
thus on average less than 2 executions of AES-256 are needed to get to Step 5.

To quantify the probability to met the condition of Step 5 we have to estimate the size of the database.
A very generous upper bound is $10^9$ PANs and $10^4$ token per PAN (generating a new token for every transaction) for a total of $10^{13}$ entries in the database.
Considering that there are $10^{16}-1$ possible tokens, the probability to met the condition is:
$$
    \rho = \frac{10^{13}}{10^{16}-1} \approx 0.001
$$
Again, the expected value of the number of iterations of the algorithm is:
$$
    \mathbb{E}_2 = \sum_{k=1}^{\infty} k \rho^{k-1} (1-\rho) \approx 1.001.
$$
thus very rarely more than one execution of SHA-256 is needed.

Finally the expected value of the total number of executions of AES-256 is:
$$
    \mathbb{E}_1 \mathbb{E}_2 \approx 1.803.
$$
thus on average less than 2 executions of AES-256 are needed to get to Step 6.

\subsection*{Acknowledgements}
The authors are indebted to several people for their suggestions:
Sandra D\'iaz, Patrick Harasser, Alessandro Tomasi and the anonymous referee. 

\bibliographystyle{plain}
\bibliography{biblio.bib%
}

\end{document}